\documentclass[11pt]{amsart}
\usepackage{amsmath,amssymb,amsthm}
\usepackage[margin=1.15in]{geometry}
\usepackage[colorlinks=true,linkcolor=blue,citecolor=blue,urlcolor=blue]{hyperref}
\hypersetup{
  pdftitle={A uniqueness theorem for exact separable log-partition decompositions},
  pdfauthor={Michael P. Rubin},
  pdfsubject={An operational characterization of the Kullback--Leibler divergence},
  pdfkeywords={Kullback--Leibler divergence, canonical divergence, log-partition function, variational free energy, dually flat geometry, characterization theorem}
}

\newtheorem{theorem}{Theorem}[section]
\newtheorem{lemma}[theorem]{Lemma}
\newtheorem{proposition}[theorem]{Proposition}
\newtheorem{corollary}[theorem]{Corollary}
\theoremstyle{definition}
\newtheorem{definition}[theorem]{Definition}
\newtheorem{example}[theorem]{Example}
\theoremstyle{remark}
\newtheorem{remark}[theorem]{Remark}

\newcommand{\X}{\mathcal{X}}
\newcommand{\simp}{\mathcal{P}^{\circ}(\X)}
\newcommand{\E}{\mathbb{E}}
\newcommand{\R}{\mathbb{R}}
\newcommand{\gap}{\widehat{\mathcal{R}}}

\begin{document}

\title[Uniqueness of separable log-partition decompositions]{A uniqueness theorem for exact separable log-partition decompositions}

\author{Michael P. Rubin}
\address{Wyss Institute for Biologically Inspired Engineering, Harvard University, Boston, MA 02115, USA}
\email{Michael.Rubin@Wyss.Harvard.edu}
\urladdr{\url{https://orcid.org/0009-0008-9814-7233}}

\subjclass[2020]{Primary 62B11; Secondary 53B12, 94A17, 39B22, 62F15, 82B03}

\keywords{Kullback--Leibler divergence; canonical divergence; log-partition function; variational free energy; dually flat geometry; characterization theorem}

\begin{abstract}
Let $\X$ be a finite set, let $p$ and $Q$ be strictly positive probability distributions on $\X$, and let $u\in\R^{\X}$ be a log-weight. Define $\Psi_{p}(u)=\log\sum_{x}p(x)e^{u(x)}$ and $\pi_{p,u}(x)=p(x)e^{u(x)-\Psi_{p}(u)}$. The variational free energy $F(Q;p,u)=D(Q\,\|\,p)-\E_{Q}[u]$ satisfies the exact identity $\Psi_{p}(u)=-F(Q;p,u)+D(Q\,\|\,\pi_{p,u})$, whose residual is the canonical divergence of the dually flat simplex. We prove a converse. Suppose $G(Q;p,u)=A(Q,p)+B(Q,u)$ is additively separable and its gap $\Psi_{p}(u)+G(Q;p,u)$ is nonnegative and vanishes whenever $Q=\pi_{p,u}$. Then $G=F$ and the gap equals $D(Q\,\|\,\pi_{p,u})$; $A$ and $B$ are unique up to a $Q$-dependent additive gauge. The gap is not assumed to be a divergence or to depend only on the target, and no continuity, measurability, differentiability, convexity, or membership in a prescribed divergence family is assumed. The proof bounds increments of a corrected accuracy term by centered log-moment-generating costs of order $k^{-2}$ per step. Telescoping $k$ equal exponential tilts gives an $O(k^{-1})$ bound and forces every increment to vanish. Thus target factorization and strictness are conclusions. Counterexamples show that separability, nonnegativity, and tightness are each necessary. Within this class, no $\alpha$-divergence distinct from Kullback--Leibler and no R\'enyi divergence of order other than one can occur as the residual.
\end{abstract}

\maketitle

\section{Introduction}\label{sec:intro}

Let $\X$ be a finite set and let $\simp$ denote the open simplex of strictly positive probability distributions on $\X$, carrying its standard dually flat information-geometric structure \cite{AmariNagaoka,Amari2016,AyJost}. Given a reference distribution $p\in\simp$ and a log-weight $u\in\R^{\X}$, exponential tilting produces
\[
\pi_{p,u}(x)=p(x)\,e^{u(x)-\Psi_{p}(u)},\qquad
\Psi_{p}(u)=\log\sum_{x\in\X}p(x)e^{u(x)},
\]
where $\Psi_{p}$ is the log-partition function of the tilt. For a trial distribution $Q\in\simp$, define the variational free energy
\begin{equation}\label{eq:F}
F(Q;p,u)=D(Q\,\|\,p)-\E_{Q}[u],
\end{equation}
where $D(Q\,\|\,R)=\sum_{x}Q(x)\log\frac{Q(x)}{R(x)}$ and $\E_{Q}[f]=\sum_{x}Q(x)f(x)$. A one-line computation, recalled as Proposition~\ref{prop:identity}, gives
\begin{equation}\label{eq:identity}
\Psi_{p}(u)=-F(Q;p,u)+D(Q\,\|\,\pi_{p,u})\qquad\text{for every }Q\in\simp.
\end{equation}
The identity expresses the log-partition function as a separated complexity and accuracy functional plus a residual, and the residual is the canonical divergence of the dually flat structure on $\simp$ \cite{AmariNagaoka,Amari2016,AyJost,AyAmari2015}.

Identity~\eqref{eq:identity} occurs in two standard settings. With $u=-\beta H$ and $p$ uniform, $F$ is the mean-field free energy $\beta\,\E_{Q}[H]-S(Q)$ up to an additive constant, and \eqref{eq:identity} together with $D(Q\,\|\,\pi_{p,u})\ge 0$ is the Gibbs--Bogoliubov inequality \cite{OpperSaad}. With $p$ a prior and $u(x)=\log L(D\mid x)$ a log-likelihood, $\Psi_{p}(u)$ is the log-evidence, $\pi_{p,u}$ the posterior, $-F$ the evidence lower bound, and \eqref{eq:identity} is the identity underlying variational Bayesian inference \cite{Jordan,WainwrightJordan,Blei}.

In applications it is common to replace the Kullback--Leibler term in a variational objective by an $\alpha$-divergence or a R\'enyi divergence, altering mass-covering and mode-seeking behaviour \cite{Minka,LiTurner}. Several such objectives furnish useful bounds on the log-partition function, but their standard formulations do not yield an exact identity of the form \eqref{eq:identity}. This raises a question that is prior to the choice among such objectives and structural in character: which functionals admit an exact decomposition of $\Psi_{p}(u)$ at all, and how much freedom is there in the choice? Since \eqref{eq:identity} is invariably obtained by defining $F$ as in \eqref{eq:F} and computing, the question does not arise in the course of the usual derivations.

The question requires care, because without structural conditions nothing is constrained: given any nonnegative function $\Phi(Q;p,u)$ that vanishes when $Q=\pi_{p,u}$, defining $G(Q;p,u)=\Phi(Q;p,u)-\Psi_{p}(u)$ makes $\Phi$ its gap (Example~\ref{ex:vacuous}). Any characterization must therefore rest on conditions governing how the computable part may depend on the model. We impose one such condition, additive separability: the computable part splits into a term depending on the reference distribution and a term depending on the log-weight, neither consulting the other. In mean-field language this asks that the trial free energy be an entropic contribution plus an energetic one; in inferential language, that the objective consist of a complexity term and an accuracy term. Beyond separability we require only that the functional majorize $-\Psi_{p}(u)$ and that the bound be exact whenever the trial distribution equals the target. The gap is not assumed to be a divergence, nor to depend on the model only through its target, and no continuity, measurability, differentiability, convexity, or membership in a prescribed divergence family is assumed. The main theorem, stated formally as Theorem~\ref{thm:main} in Section~\ref{sec:hypotheses}, asserts that these hypotheses already force the functional to be the variational free energy \eqref{eq:F} and the gap to be $D(Q\,\|\,\pi_{p,u})$; in particular, the target factorization of the residual and its strictness off the target are conclusions rather than assumptions. The split of $G$ into $A$ and $B$ is determined up to a $Q$-dependent additive gauge.

A characterization of this kind locates the freedom in a variational method. Mean-field schemes and variational inference alike proceed by fixing the functional and optimizing it over a tractable trial class, and the art of both subjects lies in the choice of that class. Theorem~\ref{thm:main} says that under its hypotheses there is no corresponding latitude at the level of the functional itself: the trial class is a modelling choice, the functional is not. The theorem also indicates where alternatives must live. An exact decomposition differing from \eqref{eq:identity} (in a quantum setting, say, or under a different reading of which part is to be computable) must relinquish one of the hypotheses, and Section~\ref{sec:sharpness} exhibits what each relinquishment permits.

\subsection*{Relation to prior work}
Relative entropy admits several characterizations, and it is useful to locate the present result among them. Shore and Johnson \cite{ShoreJohnson} axiomatize the inference procedure (which posterior should follow from a prior and a set of constraints) and conclude that one should minimize relative entropy; Csisz\'ar \cite{Csiszar1991} gives a related axiomatic treatment for linear inverse problems, and Uffink \cite{Uffink} later examined the uniqueness claim in such arguments. Csisz\'ar \cite{Csiszar1967} and Ali and Silvey \cite{AliSilvey} develop the class of $f$-divergences, Bregman \cite{Bregman} introduces the class bearing his name, and Amari \cite{Amari2009} shows that the Kullback--Leibler divergence lies in both classes and is alone in doing so among divergences on probability distributions (see also \cite{AmariNagaoka}). Within information geometry the Kullback--Leibler divergence arises in a further way: it is the canonical divergence of the dually flat structure on the simplex \cite{AmariNagaoka,Amari2016,AyJost}, and Ay and Amari \cite{AyAmari2015} construct a canonical divergence from geodesic data that recovers it in the dually flat case; Nielsen \cite{Nielsen2020} surveys this circle of ideas. A neighboring literature characterizes proper local scoring rules on discrete sample spaces \cite{DawidLauritzenParry}; there the primitive object is an outcome-dependent scoring rule, whereas here no scoring rule is assumed and the primitive object is an exact decomposition of a log-normalizer under separation of its reference-distribution and log-weight arguments. Recent work has also examined the information geometry of evidence-lower-bound optimization, including the relation between the Fisher--Rao natural gradient of the evidence lower bound and Kullback--Leibler minimization on constrained models \cite{AyVanOostrumDatar}; the present question is logically prior to model restriction and optimization, since it characterizes which exact decomposition underlies the objective in the first place.

Theorem~\ref{thm:main} is of a different type from the divergence characterizations above. It constrains what a decomposition may be, and the divergence is an output rather than an input: the gap is not postulated to be a divergence, no geometric structure is assumed, and none of the cited characterizations is invoked. A 2026 preprint by Takeo Imaizumi, ``An inverse variational characterization of the Kullback--Leibler divergence'' (Jxiv 2883, \href{https://doi.org/10.51094/jxiv.2883}{doi:10.51094/jxiv.2883}), studies a distinct one-dimensional problem under $C^{2}$, homogeneity, normalization, and strict-convexity assumptions; it does not address separable log-partition decompositions. To the best of our knowledge, no previous result characterizes the decompositions considered here. The theorem is silent about which objective to optimize when an exact decomposition is not required, which is the setting in which the alternative divergences of \cite{Minka,LiTurner} are proposed and are useful; its scope is the class of exact separable decompositions. The relation to the canonical divergence of dually flat geometry is taken up in Remark~\ref{rem:canonical}.

\subsection*{Organization}
Section~\ref{sec:setting} fixes notation and records the geometric meaning of $\Psi_{p}$. Section~\ref{sec:decomposition} recalls \eqref{eq:identity}. Section~\ref{sec:hypotheses} states the hypotheses and the theorem, Section~\ref{sec:proof} proves it, and Section~\ref{sec:sharpness} shows by counterexample that each hypothesis is needed. Section~\ref{sec:discussion} discusses the relation to convex duality and to the canonical divergence, and the extension beyond finite $\X$.

\section{Setting}\label{sec:setting}

Throughout, $\X$ is a finite set with $|\X|=n\ge 2$, and
\[
\simp=\Bigl\{Q\in\R^{\X}:Q(x)>0\ \forall x,\ \textstyle\sum_{x}Q(x)=1\Bigr\}.
\]
We call $p\in\simp$ a reference distribution, $u\in\R^{\X}$ a log-weight, and the pair $(p,u)\in\simp\times\R^{\X}$ a model. For a model we write
\[
\Psi_{p}(u)=\log\sum_{x}p(x)e^{u(x)},\qquad
\pi_{p,u}=p\,e^{u-\Psi_{p}(u)}\in\simp,
\]
and call $\pi_{p,u}$ the target of the model. For $Q,R\in\simp$ set $D(Q\,\|\,R)=\sum_{x}Q(x)\log\frac{Q(x)}{R(x)}$; recall $D(Q\,\|\,R)\ge 0$ with equality iff $Q=R$ \cite{CoverThomas}. For $f,m\in\R^{\X}$ we write $\langle m,f\rangle=\sum_{x}m(x)f(x)$, and $\mathbf{1}$ for the all-ones vector. The variational free energy is defined by \eqref{eq:F}.

Two elementary observations are used repeatedly. First, the models $(p,u)$ and $(p,u+c\mathbf{1})$ have the same target for every $c\in\R$, while $\Psi_{p}(u+c\mathbf{1})=\Psi_{p}(u)+c$: constants in the log-weight are pure normalization. Second, $(p,u)$ and $(p',u')$ have the same target if and only if $\log p+u$ and $\log p'+u'$ differ by a constant vector.

\begin{remark}[Geometric meaning of $\Psi_{p}$]\label{rem:geometry}
For fixed $p$, the function $\Psi_{p}$ is smooth and convex on $\R^{\X}$, with gradient and Hessian
\[
\nabla\Psi_{p}(u)=\pi_{p,u},\qquad
\nabla^{2}\Psi_{p}(u)=\operatorname{diag}(\pi_{p,u})-\pi_{p,u}\pi_{p,u}^{\top},
\]
the latter being the covariance form under $\pi_{p,u}$, positive semidefinite with kernel the constant vectors. After quotienting by constants, the log-weights are exponential (natural) coordinates for $\simp$, the targets are the dual expectation coordinates, and the Hessian induces the Fisher--Rao metric. The Bregman divergence generated by $\Psi_{p}$ is the Kullback--Leibler divergence between tilted targets, with the orientation
\[
B_{\Psi_{p}}(u,v)=\Psi_{p}(u)-\Psi_{p}(v)-\langle\nabla\Psi_{p}(v),\,u-v\rangle
=D(\pi_{p,v}\,\|\,\pi_{p,u}),
\]
the canonical divergence of the dually flat structure \cite{AmariNagaoka,Amari2016,AyJost}. Theorem~\ref{thm:main} characterizes the same divergence without invoking any of this structure: no metric, no connection, and no convexity appear among its hypotheses.
\end{remark}

\section{The decomposition identity}\label{sec:decomposition}

\begin{proposition}\label{prop:identity}
For every model $(p,u)$ and every $Q\in\simp$,
\[
\Psi_{p}(u)=-F(Q;p,u)+D(Q\,\|\,\pi_{p,u}).
\]
\end{proposition}

\begin{proof}
Write $\pi=\pi_{p,u}$ and $\Psi=\Psi_{p}(u)$. Since $\log\pi=\log p+u-\Psi\mathbf{1}$,
\[
-F(Q;p,u)+D(Q\,\|\,\pi)
=\langle Q,\log p-\log Q\rangle+\langle Q,u\rangle+\langle Q,\log Q-\log\pi\rangle
=\langle Q,\Psi\mathbf{1}\rangle=\Psi,
\]
using $\langle Q,\mathbf{1}\rangle=1$.
\end{proof}

Since $D(Q\,\|\,\pi)\ge 0$ with equality iff $Q=\pi$, Proposition~\ref{prop:identity} yields at once the variational principle $F(Q;p,u)\ge -\Psi_{p}(u)$, with equality exactly at $Q=\pi_{p,u}$; this is the Gibbs--Bogoliubov inequality in the present notation.

\section{Hypotheses and statement}\label{sec:hypotheses}

\begin{definition}[Gap]\label{def:gap}
For a function $G:\simp\times\simp\times\R^{\X}\to\R$, define its gap
\[
\gap(Q;p,u):=\Psi_{p}(u)+G(Q;p,u).
\]
\end{definition}

By construction, $\Psi_{p}(u)=-G(Q;p,u)+\gap(Q;p,u)$ identically: every $G$ decomposes the log-partition function exactly, with its own gap as residual. The content of a characterization therefore lies entirely in structural conditions. We impose three.

\begin{itemize}
\item[(S)] \textbf{Separability.} There exist $A:\simp\times\simp\to\R$ and $B:\simp\times\R^{\X}\to\R$ with $G(Q;p,u)=A(Q,p)+B(Q,u)$ for all $Q,p,u$.
\item[(N)] \textbf{Nonnegativity.} $\gap(Q;p,u)\ge 0$ for all $Q,p,u$; equivalently, $G$ majorizes the negative log-partition function, $G(Q;p,u)\ge-\Psi_{p}(u)$.
\item[(T)] \textbf{Tightness at the target.} $\gap(Q;p,u)=0$ whenever $Q=\pi_{p,u}$; the bound in (N) is exact when the trial distribution equals the target.
\end{itemize}

Hypothesis (S) is the substantive one. It asks that the computable part be additive in the reference distribution and the log-weight, each term depending on its own argument alone. It does not name the terms; it restricts them from consulting one another. Hypotheses (N) and (T) say only that $G$ is an upper bound on $-\Psi_{p}(u)$ that is attained at the exponential target. Nothing else is assumed of the gap: it is not assumed to be a divergence, it is not assumed to vanish only at the target, and it is not assumed to depend on the model through the pair $(Q,\pi_{p,u})$ alone.

\begin{theorem}\label{thm:main}
Let $G$ satisfy (S), (N) and (T). Then
\[
G(Q;p,u)=F(Q;p,u)=D(Q\,\|\,p)-\E_{Q}[u],\qquad
\gap(Q;p,u)=D(Q\,\|\,\pi_{p,u})
\]
for all $Q,p\in\simp$ and $u\in\R^{\X}$. Moreover, every pair $A,B$ witnessing (S) is determined up to an additive gauge: necessarily
\[
A(Q,p)=D(Q\,\|\,p)+a(Q),\qquad B(Q,u)=-\E_{Q}[u]-a(Q)
\]
for some $a:\simp\to\R$, which cancels in $G$.
\end{theorem}

The theorem is proved in Section~\ref{sec:proof}. Three features of the conclusion deserve notice. First, the gap depends on the model only through the pair $(Q,\pi_{p,u})$: the target factorization of the residual, which one might have thought must be postulated, is derived. Second, the gap vanishes only at $Q=\pi_{p,u}$: strictness is likewise derived. Third, no regularity hypothesis appears anywhere. Nonlinear additive pathologies of Cauchy type, which ordinarily force a measurability, continuity, or boundedness hypothesis in characterization results of this kind \cite{Aczel}, never arise: the telescoping argument of the proof bounds increments directly.

The formulation in which the residual is postulated to be a function of $(Q,\pi)$ alone is the special case in which $\gap(Q;p,u)=\mathcal{R}(Q,\pi_{p,u})$ for some $\mathcal{R}:\simp\times\simp\to\R$; since every $\pi\in\simp$ is realized as a target (take $p=\pi$, $u=0$), Theorem~\ref{thm:main} then gives $\mathcal{R}=D(\cdot\,\|\,\cdot)$ on all of $\simp\times\simp$.

\begin{corollary}\label{cor:exclusion}
Let $\mathcal{R}_{0}:\simp\times\simp\to\R$ be any function with $\mathcal{R}_{0}\ne D(\cdot\,\|\,\cdot)$; for instance an $\alpha$-divergence distinct from $D(\cdot\,\|\,\cdot)$, or a R\'enyi divergence of order $r\ne 1$. Then there is no $G$ satisfying (S) whose gap is $\mathcal{R}_{0}(Q,\pi_{p,u})$ and which satisfies (N) and (T). Equivalently: within the class of additively separable exact decompositions of the log-partition function with nonnegative residual vanishing at the target, no residual other than the Kullback--Leibler divergence can occur.
\end{corollary}

\begin{proof}
Immediate from Theorem~\ref{thm:main}: any such $G$ has gap $D(Q\,\|\,\pi_{p,u})$, and since every $\pi\in\simp$ is a target, $\mathcal{R}_{0}=D(\cdot\,\|\,\cdot)$, a contradiction.
\end{proof}

Corollary~\ref{cor:exclusion} bears on a practice rather than only on a formalism. Several $\alpha$- and R\'enyi-based objectives furnish useful bounds on the log-partition function \cite{Minka,LiTurner}; the corollary says that, within the class of separable exact decompositions with nonnegative residual vanishing at the target, such a bound cannot also have an exact residual of the prescribed divergent form without relinquishing at least one hypothesis of the theorem. Thus a useful variational bound and the exact identity characterized here are distinct structural properties.

\section{Proof of Theorem~\ref{thm:main}}\label{sec:proof}

\begin{proof}[Proof of Theorem~{\upshape\ref{thm:main}}]
Fix $G$ satisfying (S), (N) and (T), with summands $A,B$ as in (S).

\smallskip
\emph{Step 1: reduction.} Define
\[
A'(Q,p):=A(Q,p)-D(Q\,\|\,p),\qquad B'(Q,u):=B(Q,u)+\E_{Q}[u].
\]
By Proposition~\ref{prop:identity},
\begin{equation}\label{eq:residual}
\gap(Q;p,u)=\Psi_{p}(u)+G(Q;p,u)=D(Q\,\|\,\pi_{p,u})+A'(Q,p)+B'(Q,u).
\end{equation}
We show that for each $Q$, $A'(Q,\cdot)$ and $B'(Q,\cdot)$ are constant and sum to zero, which gives the theorem. Fix $Q\in\simp$ for the remainder of the proof.

\smallskip
\emph{Step 2: tightness along tilts of $Q$.} For $u\in\R^{\X}$, define
\[
p_{u}:=\frac{Q\,e^{-u}}{\E_{Q}[e^{-u}]}\in\simp.
\]
Then $p_{u}e^{u}\propto Q$, so $\pi_{p_{u},u}=Q$, and since $D(Q\,\|\,Q)=0$, hypothesis (T) applied to the model $(p_{u},u)$ in \eqref{eq:residual} gives
\begin{equation}\label{eq:tight}
A'(Q,p_{u})+B'(Q,u)=0\qquad\text{for every }u\in\R^{\X}.
\end{equation}
Two consequences follow. First, since $p_{u+c\mathbf{1}}=p_{u}$ for every $c\in\R$, comparing \eqref{eq:tight} at $u$ and at $u+c\mathbf{1}$ yields
\begin{equation}\label{eq:const}
B'(Q,u+c\mathbf{1})=B'(Q,u)\qquad\text{for all }c\in\R,\ u\in\R^{\X}:
\end{equation}
constants in the log-weight are invisible to $B'$. Second, the map $u\mapsto p_{u}$ is onto $\simp$: given $p\in\simp$, the log-weight $a_{p}:=\log Q-\log p$ satisfies $p\,e^{a_{p}}=Q$, hence $\E_{Q}[e^{-a_{p}}]=\sum_{x}p(x)=1$ and $p_{a_{p}}=p$. Therefore \eqref{eq:tight} determines $A'$ from $B'$:
\begin{equation}\label{eq:Adet}
A'(Q,p)=-B'(Q,a_{p})=-B'\bigl(Q,\log Q-\log p\bigr)\qquad\text{for every }p\in\simp.
\end{equation}

\smallskip
\emph{Step 3: nonnegativity gives a two-sided increment bound.} For arbitrary $p\in\simp$ and $u\in\R^{\X}$, write $a=a_{p}$; then $p\,e^{u}=Q\,e^{u-a}$, so
\[
\pi_{p,u}=\frac{Q\,e^{u-a}}{\E_{Q}[e^{u-a}]},\qquad
D(Q\,\|\,\pi_{p,u})=\log\E_{Q}[e^{u-a}]-\E_{Q}[u-a]=:K_{Q}(u-a),
\]
the centered log-moment-generating function of $u-a$ under $Q$. Substituting \eqref{eq:Adet} into \eqref{eq:residual}, hypothesis (N) reads
\[
0\ \le\ \gap(Q;p,u)=K_{Q}(u-a)+B'(Q,u)-B'(Q,a).
\]
As $p$ ranges over $\simp$, the vector $a=a_{p}$ ranges over
$\{a\in\R^{\X}:\E_{Q}[e^{-a}]=1\}$. For an arbitrary $a\in\R^{\X}$, set
\[
c_{a}:=\log\E_{Q}[e^{-a}],\qquad \widetilde a:=a+c_{a}\mathbf{1}.
\]
Then $\E_{Q}[e^{-\widetilde a}]=1$ and $\widetilde a=a_{p}$ for
$p=Qe^{-a}/\E_{Q}[e^{-a}]\in\simp$. Moreover,
$B'(Q,\widetilde a)=B'(Q,a)$ by \eqref{eq:const}, and
$K_{Q}(u-\widetilde a)=K_{Q}(u-a)$ because
$K_{Q}(w+c\mathbf{1})=K_{Q}(w)$. Hence
\[
B'(Q,u)-B'(Q,a)\ \ge\ -K_{Q}(u-a)\qquad\text{for all }a,u\in\R^{\X}.
\]
Interchanging the roles of $a$ and $u$ gives the matching upper bound, so
\begin{equation}\label{eq:increment}
-K_{Q}(u-a)\ \le\ B'(Q,u)-B'(Q,a)\ \le\ K_{Q}(a-u)\qquad\text{for all }a,u\in\R^{\X}.
\end{equation}

\smallskip
\emph{Step 4: telescoping kills the increments.} The function $K_{Q}$ satisfies, for every $w\in\R^{\X}$,
\begin{equation}\label{eq:Kbound}
0\ \le\ K_{Q}(w)\ \le\ \E_{Q}[e^{w}]-1-\E_{Q}[w]\ \le\ \tfrac{1}{2}\,\|w\|_{\infty}^{2}\,e^{\|w\|_{\infty}}:
\end{equation}
the first inequality is Jensen's, the second is $\log y\le y-1$, and the third is Taylor's theorem with Lagrange remainder, $e^{s}\le 1+s+\tfrac{s^{2}}{2}e^{|s|}$. Now fix $a,v\in\R^{\X}$ and an integer $k\ge 1$, and set $a_{j}:=a+\tfrac{j}{k}v$ for $j=0,\dots,k$. Applying \eqref{eq:increment} to each consecutive pair $(a_{j-1},a_{j})$, whose difference is $v/k$, and summing the telescoping increments,
\[
-k\,K_{Q}(v/k)\ \le\ B'(Q,a+v)-B'(Q,a)\ \le\ k\,K_{Q}(-v/k).
\]
By \eqref{eq:Kbound}, with
\[
C_{k}:=\frac{\|v\|_{\infty}^{2}}{2k}\,e^{\|v\|_{\infty}/k},
\]
we have
\[
-C_{k}\ \le\ B'(Q,a+v)-B'(Q,a)\ \le\ C_{k},
\qquad
0\le C_{k}\le \frac{\|v\|_{\infty}^{2}}{2k}e^{\|v\|_{\infty}}\longrightarrow 0.
\]
Therefore $B'(Q,a+v)=B'(Q,a)$ for all $a,v\in\R^{\X}$, so $B'(Q,\cdot)$ is constant.

\smallskip
\emph{Step 5: conclusion.} Write $B'(Q,u)=-a(Q)$ for all $u$. By \eqref{eq:Adet}, $A'(Q,p)=a(Q)$ for all $p$. Therefore
\[
A(Q,p)=D(Q\,\|\,p)+a(Q),\qquad B(Q,u)=-\E_{Q}[u]-a(Q),
\]
\[
G(Q;p,u)=A(Q,p)+B(Q,u)=D(Q\,\|\,p)-\E_{Q}[u]=F(Q;p,u),
\]
and $\gap(Q;p,u)=D(Q\,\|\,\pi_{p,u})$ by \eqref{eq:residual}. Conversely, every pair $(A,B)$ of the displayed form satisfies (S) and reproduces $G=F$, so the gauge $a$ is exactly the residual freedom.
\end{proof}

\begin{remark}[The mechanism]\label{rem:mechanism}
The proof turns on a mismatch in scaling along exponential tilts, which are the $e$-geodesics through $Q$ in the language of \cite{AmariNagaoka}, although the proof nowhere uses this interpretation. Splitting a tilt into $k$ equal steps multiplies the number of increments by $k$ exactly, an arithmetic fact requiring nothing of $B'$, while the tilt cost per step, $K_{Q}(\pm v/k)$, is of order $k^{-2}$. Inequality \eqref{eq:increment} therefore pins every increment of $B'$ between sequences that vanish. This is why no regularity hypothesis is needed: the pathological additive functionals that ordinarily force one in characterization results of this kind \cite{Aczel} never enter, because the argument bounds $B'$ directly rather than first extracting a functional equation for it.
\end{remark}

\section{Sharpness of the hypotheses}\label{sec:sharpness}

Each of (S), (N), (T) is needed, and without (S) the remaining hypotheses are vacuous.

\begin{example}[(N) and (T) alone are vacuous]\label{ex:vacuous}
Let $\Phi:\simp\times\simp\times\R^{\X}\to\R$ be any function with $\Phi\ge 0$ and $\Phi(Q;p,u)=0$ whenever $Q=\pi_{p,u}$, and define $G(Q;p,u):=\Phi(Q;p,u)-\Psi_{p}(u)$. Then $G$ has gap $\Phi$ and satisfies (N) and (T). Hence no uniqueness can follow from (N) and (T) without further conditions, and uniqueness arises only from their interaction with separability.
\end{example}

\begin{lemma}\label{lem:nonsep}
There are no functions $\alpha:\simp\to\R$ and $\beta:\R^{\X}\to\R$ with $\Psi_{p}(u)=\alpha(p)+\beta(u)$ for all $(p,u)$.
\end{lemma}

\begin{proof}
Suppose there were. Taking $u=0$ gives $\Psi_{p}(0)=\log\sum_{x}p(x)=0$, so $\alpha(p)=-\beta(0)$ for every $p$, and $\alpha$ is a constant. Then $\Psi_{p}(u)$ is independent of $p$ for every $u$. But since $n\ge 2$ we may fix $x_{1}\in\X$ and take $u=t\,\delta_{x_{1}}$ with $t\ne 0$, for which $\Psi_{p}(u)=\log\bigl(1+p(x_{1})(e^{t}-1)\bigr)$ is nonconstant in $p$, a contradiction.
\end{proof}

\begin{example}[(S) is necessary]\label{ex:S}
Take $\Phi(Q;p,u):=2\,D(Q\,\|\,\pi_{p,u})$ and define $G(Q;p,u):=\Phi(Q;p,u)-\Psi_{p}(u)$ as in Example~\ref{ex:vacuous}. Then (N) and (T) hold, since $2D(\cdot\,\|\,\cdot)$ is nonnegative and vanishes at the target; yet the gap is not $D(Q\,\|\,\pi_{p,u})$ and, by Proposition~\ref{prop:identity}, $G\ne F$. Hypothesis (S) must therefore fail, and it does. By Proposition~\ref{prop:identity},
\[
G(Q;p,u)=2\,D(Q\,\|\,p)-2\,\E_{Q}[u]+\Psi_{p}(u),
\]
so if $G(Q;p,u)=A(Q,p)+B(Q,u)$ held, then fixing any single $Q_{0}\in\simp$ and putting $\alpha(p):=A(Q_{0},p)-2\,D(Q_{0}\,\|\,p)$ and $\beta(u):=B(Q_{0},u)+2\,\E_{Q_{0}}[u]$ would give $\Psi_{p}(u)=\alpha(p)+\beta(u)$, contradicting Lemma~\ref{lem:nonsep}. More generally, for any scalar $\lambda\ne 1$, the functional associated with the gap $\lambda D(\cdot\,\|\,\cdot)$ is $G_{\lambda}=\lambda F+(\lambda-1)\Psi_{p}(u)$; if $G_{\lambda}$ were separable, then $\Psi_{p}(u)$ would be separable, contrary to Lemma~\ref{lem:nonsep}. For $\lambda\ge 0$, these gaps also satisfy (N) and (T).
\end{example}

\begin{example}[(T) is necessary]\label{ex:tight}
Fix $c>0$ and take $G:=F+c$, with $A(Q,p)=D(Q\,\|\,p)+c$ and $B(Q,u)=-\E_{Q}[u]$. Then $G$ is separable and its gap is $D(Q\,\|\,\pi_{p,u})+c\ge c>0$, so (N) holds; but the gap equals $c\ne 0$ at $Q=\pi_{p,u}$, so (T) fails, and $G\ne F$.
\end{example}

\begin{example}[(N) is necessary]\label{ex:nonneg}
Fix $m\in\R^{\X}$ with $\sum_{x}m(x)=0$ and $m\ne 0$, and set
\begin{align*}
G_{m}(Q;p,u)&:=F(Q;p,u)+\langle m,\log p\rangle+\langle m,u\rangle-\langle m,\log Q\rangle,
\end{align*}
which is separable, with $A(Q,p)=D(Q\,\|\,p)+\langle m,\log p\rangle-\langle m,\log Q\rangle$ and $B(Q,u)=-\E_{Q}[u]+\langle m,u\rangle$. Since $\sum_{x}m(x)=0$ we have $\langle m,\log p+u\rangle=\langle m,\log\pi_{p,u}\rangle$, so the gap of $G_{m}$ is
\begin{align*}
\gap_{m}(Q;p,u)&=D(Q\,\|\,\pi_{p,u})+\Bigl\langle m,\log\frac{\pi_{p,u}}{Q}\Bigr\rangle,
\end{align*}
which vanishes at $Q=\pi_{p,u}$, so (T) holds. But (N) fails. Since $m\ne 0$ and $\sum_{x}m(x)=0$, $m$ is not a multiple of $Q$; hence there is $g\in\R^{\X}$ with $\langle Q,g\rangle=0$ and $\langle m,g\rangle\ne 0$. For sufficiently small $|\varepsilon|$, the vector $\pi_{\varepsilon}:=Q(1+\varepsilon g)$ belongs to $\simp$. Taking the model $(p,u)=(\pi_{\varepsilon},0)$, whose target is $\pi_{\varepsilon}$, gives
\[
D(Q\,\|\,\pi_{\varepsilon})=-\E_{Q}[\log(1+\varepsilon g)]=\frac{\varepsilon^{2}}{2}\E_{Q}[g^{2}]+O(\varepsilon^{3})
\]
and
\[
\Bigl\langle m,\log\frac{\pi_{\varepsilon}}{Q}\Bigr\rangle
=\langle m,\log(1+\varepsilon g)\rangle
=\varepsilon\langle m,g\rangle+O(\varepsilon^{2}).
\]
Thus $\gap_{m}(Q;\pi_{\varepsilon},0)=\varepsilon\langle m,g\rangle+O(\varepsilon^{2})$, which is negative for sufficiently small $\varepsilon$ of the appropriate sign. Concretely, with $\X=\{1,2\}$ and $m=(1,-1)$ the gap is $D(Q\,\|\,\pi)+\log\frac{\pi_{1}Q_{2}}{\pi_{2}Q_{1}}$ with $\pi=\pi_{p,u}$. Thus $G_{m}\ne F$.
\end{example}

\section{Discussion}\label{sec:discussion}

\begin{remark}[Relation to convex duality]\label{rem:duality}
Identity \eqref{eq:identity} is closely related to the Gibbs variational principle
\[
\Psi_{p}(u)=\sup_{Q\in\simp}\Bigl\{\E_{Q}[u]-D(Q\,\|\,p)\Bigr\}
=-\inf_{Q\in\simp}F(Q;p,u),
\]
whose supremum is attained at $Q=\pi_{p,u}$. Equivalently, \eqref{eq:identity} states that
\[
F(Q;p,u)-F(\pi_{p,u};p,u)=D(Q\,\|\,\pi_{p,u}).
\]
This is the Bregman divergence generated by the negative entropy in the trial variable. The theorem is not an immediate consequence of Fenchel duality. Duality identifies a conjugate once one member of the pair is prescribed, whereas the present result characterizes which separated functionals can participate in the exact decomposition at all. Nor does the Bregman reading by itself yield uniqueness: every strictly convex generator has a strict Bregman divergence. The present theorem instead imposes the decomposition-level conditions (S), (N) and (T), without prescribing a convex generator or a conjugate pair. Example~\ref{ex:S} makes the distinction concrete: its $G$ is convex in $Q$ and its gap is the strict divergence $2D(Q\,\|\,\pi_{p,u})$, yet separability fails.
\end{remark}

\begin{remark}[Relation to the canonical divergence]\label{rem:canonical}
As recalled in Remark~\ref{rem:geometry}, the simplex carries the dually flat structure of Amari and Nagaoka, and the Kullback--Leibler divergence is its canonical divergence \cite{AmariNagaoka,Amari2016,AyJost}; Ay and Amari \cite{AyAmari2015} obtain it from geodesic data by a construction available on general dualistic manifolds. These results single out $D(\cdot\,\|\,\cdot)$ through geometric input: a metric, a pair of affine connections, and the compatibility between them. Theorem~\ref{thm:main} singles out the same divergence by a route in which no geometric structure appears: the hypotheses concern only the algebra of the decomposition (separability, nonnegativity, tightness at the target), and neither convexity nor any metric enters; even the fact that the residual factors through the target is derived rather than imposed. That an operational requirement on the decomposition of the log-partition function and a geometric requirement on the dualistic structure of the simplex select the same object is a consistency that neither result implies of the other. The theorem may accordingly be read as an operational counterpart, within the setting of identity \eqref{eq:identity}, of the canonical status the Kullback--Leibler divergence enjoys in the dually flat geometry of the simplex, and Corollary~\ref{cor:exclusion} as the corresponding exclusion: divergences that are geometrically natural on other grounds, the $\alpha$-divergences among them \cite{Amari2009,Amari2016}, cannot inherit the exact decomposition.
\end{remark}

\begin{remark}[The gauge freedom]\label{rem:gauge}
Theorem~\ref{thm:main} determines $G$ and the gap but not $A$ and $B$ individually: the function $a(Q)$ is free and cancels. Only the sum enters the decomposition, so no condition stated in terms of $G$ can fix the split. The natural normalization $A(Q,p)=D(Q\,\|\,p)$, i.e.\ $a\equiv 0$, is the one under which the entropic term vanishes at $Q=p$.
\end{remark}

\begin{remark}[Beyond finite $\X$]\label{rem:beyond}
An extension to general measure spaces requires specifying an admissible class of densities and log-weights that is closed under the bounded exponential tilts used in the proof. On such a class, the constructions $p_{u}=Qe^{-u}/\E_{Q}[e^{-u}]$ and $a+\tfrac{j}{k}v$, together with the elementary bound \eqref{eq:Kbound}, suggest a direct extension for bounded log-density ratios and log-weights. We leave the precise functional-analytic formulation to future work.
\end{remark}

\begin{remark}[Interpretation in the two readings]\label{rem:readings}
In the mean-field reading, Theorem~\ref{thm:main} states that, under (S), (N) and (T), the Gibbs--Bogoliubov functional is the unique entropy-plus-energy functional that bounds $-\log Z$ from above with equality at the Gibbs distribution. In the inferential reading, under the same assumptions, the evidence lower bound is the unique complexity-plus-accuracy lower bound on the log-evidence that is tight at the posterior. In particular, under the hypotheses of Theorem~\ref{thm:main} the Kullback--Leibler divergence occurring in \eqref{eq:identity} is determined rather than selected.
\end{remark}


\begin{thebibliography}{99}

\bibitem{Aczel}
Acz\'el, J.: Lectures on Functional Equations and Their Applications. Academic Press, New York (1966)

\bibitem{AliSilvey}
Ali, S.M., Silvey, S.D.: A general class of coefficients of divergence of one distribution from another. J. R. Stat. Soc. Ser. B \textbf{28}, 131--142 (1966). \url{https://doi.org/10.1111/j.2517-6161.1966.tb00626.x}

\bibitem{Amari2009}
Amari, S.: $\alpha$-divergence is unique, belonging to both $f$-divergence and Bregman divergence classes. IEEE Trans. Inf. Theory \textbf{55}(11), 4925--4931 (2009). \url{https://doi.org/10.1109/TIT.2009.2030485}

\bibitem{Amari2016}
Amari, S.: Information Geometry and Its Applications. Applied Mathematical Sciences, vol. 194. Springer, Tokyo (2016). \url{https://doi.org/10.1007/978-4-431-55978-8}

\bibitem{AmariNagaoka}
Amari, S., Nagaoka, H.: Methods of Information Geometry. Translations of Mathematical Monographs, vol. 191. American Mathematical Society, Providence, and Oxford University Press, Oxford (2000)

\bibitem{AyAmari2015}
Ay, N., Amari, S.: A novel approach to canonical divergences within information geometry. Entropy \textbf{17}(12), 8111--8129 (2015). \url{https://doi.org/10.3390/e17127866}

\bibitem{AyJost}
Ay, N., Jost, J., L\^e, H.V., Schwachh\"ofer, L.: Information Geometry. Ergebnisse der Mathematik und ihrer Grenzgebiete. 3. Folge, vol. 64. Springer, Cham (2017). \url{https://doi.org/10.1007/978-3-319-56478-4}

\bibitem{AyVanOostrumDatar}
Ay, N., van Oostrum, J., Datar, A.: On the natural gradient of the evidence lower bound. J. Mach. Learn. Res. \textbf{26}(222), 1--37 (2025). \url{https://jmlr.org/papers/v26/24-0606.html}

\bibitem{Blei}
Blei, D.M., Kucukelbir, A., McAuliffe, J.D.: Variational inference: a review for statisticians. J. Am. Stat. Assoc. \textbf{112}(518), 859--877 (2017). \url{https://doi.org/10.1080/01621459.2017.1285773}

\bibitem{Bregman}
Bregman, L.M.: The relaxation method of finding the common point of convex sets and its application to the solution of problems in convex programming. USSR Comput. Math. Math. Phys. \textbf{7}, 200--217 (1967). \url{https://doi.org/10.1016/0041-5553(67)90040-7}

\bibitem{CoverThomas}
Cover, T.M., Thomas, J.A.: Elements of Information Theory, 2nd edn. Wiley, Hoboken (2006)

\bibitem{Csiszar1967}
Csisz\'ar, I.: Information-type measures of difference of probability distributions and indirect observations. Studia Sci. Math. Hungar. \textbf{2}, 299--318 (1967)

\bibitem{Csiszar1991}
Csisz\'ar, I.: Why least squares and maximum entropy? An axiomatic approach to inference for linear inverse problems. Ann. Stat. \textbf{19}(4), 2032--2066 (1991). \url{https://doi.org/10.1214/aos/1176348385}

\bibitem{DawidLauritzenParry}
Dawid, A.P., Lauritzen, S., Parry, M.: Proper local scoring rules on discrete sample spaces. Ann. Stat. \textbf{40}(1), 593--608 (2012). \url{https://doi.org/10.1214/12-AOS972}

\bibitem{Jordan}
Jordan, M.I., Ghahramani, Z., Jaakkola, T.S., Saul, L.K.: An introduction to variational methods for graphical models. Mach. Learn. \textbf{37}, 183--233 (1999). \url{https://doi.org/10.1023/A:1007665907178}

\bibitem{LiTurner}
Li, Y., Turner, R.E.: R\'enyi divergence variational inference. In: Advances in Neural Information Processing Systems 29, pp. 1073--1081 (2016)

\bibitem{Minka}
Minka, T.: Divergence measures and message passing. Technical Report MSR-TR-2005-173, Microsoft Research (2005)

\bibitem{Nielsen2020}
Nielsen, F.: An elementary introduction to information geometry. Entropy \textbf{22}(10), 1100 (2020). \url{https://doi.org/10.3390/e22101100}

\bibitem{OpperSaad}
Opper, M., Saad, D. (eds.): Advanced Mean Field Methods: Theory and Practice. MIT Press, Cambridge (2001)

\bibitem{ShoreJohnson}
Shore, J.E., Johnson, R.W.: Axiomatic derivation of the principle of maximum entropy and the principle of minimum cross-entropy. IEEE Trans. Inf. Theory \textbf{26}(1), 26--37 (1980). \url{https://doi.org/10.1109/TIT.1980.1056144}

\bibitem{Uffink}
Uffink, J.: Can the maximum entropy principle be explained as a consistency requirement? Stud. Hist. Philos. Mod. Phys. \textbf{26}(3), 223--261 (1995). \url{https://doi.org/10.1016/1355-2198(95)00015-1}

\bibitem{WainwrightJordan}
Wainwright, M.J., Jordan, M.I.: Graphical models, exponential families, and variational inference. Found. Trends Mach. Learn. \textbf{1}(1--2), 1--305 (2008). \url{https://doi.org/10.1561/2200000001}

\end{thebibliography}
\end{document}